\documentclass[12pt]{article}

\usepackage{latexsym,amsmath,amscd,amssymb,graphics,float}
\usepackage{enumerate}

\usepackage{graphicx,lscape}

\usepackage[colorlinks]{hyperref}
\usepackage{url}

\begin{document}

\newenvironment{proof}[1][Proof]{\textbf{#1.} }{\ \rule{0.5em}{0.5em}}

\newtheorem{theorem}{Theorem}[section]
\newtheorem{definition}[theorem]{Definition}
\newtheorem{lemma}[theorem]{Lemma}
\newtheorem{remark}[theorem]{Remark}
\newtheorem{proposition}[theorem]{Proposition}
\newtheorem{corollary}[theorem]{Corollary}
\newtheorem{example}[theorem]{Example}

\numberwithin{equation}{section}
\newcommand{\ep}{\varepsilon}
\newcommand{\R}{{\mathbb  R}}
\newcommand\C{{\mathbb  C}}
\newcommand\Q{{\mathbb Q}}
\newcommand\Z{{\mathbb Z}}
\newcommand{\N}{{\mathbb N}}

\newcommand{\bfi}{\bfseries\itshape}

\newsavebox{\savepar}
\newenvironment{boxit}{\begin{lrbox}{\savepar}
\begin{minipage}[b]{15.5cm}}{\end{minipage}\end{lrbox}
\fbox{\usebox{\savepar}}}

\title{{\bf On local generators of affine distributions on Riemannian manifolds}}
\author{R\u{a}zvan M. Tudoran}

\date{}
\maketitle \makeatother

\begin{abstract}
Using a coordinate free characterization of hyperplanes intersection, we provide explicitly a set of local generators for a smooth affine distribution given by those smooth vector fields $X\in\mathfrak{X}(U)$ defined eventually on an open subset $U\subseteq M$ of a smooth Riemannian manifold $(M,g)$, that verifies the relations $g(X,X_1)=\dots =g(X,X_k)=0$, $g(X,Y_1)=h_1, \dots, g(X,Y_p)=h_p$, where $X_1,\dots, X_k, Y_1, \dots, Y_p \in\mathfrak{X}(U)$, and respectively $h_1,\dots, h_p \in \mathcal{C}^{\infty}(U,\mathbb{R})$, are a-priori given quantities. 

In the case when $X_1,\dots, X_k, Y_1, \dots, Y_p$ are gradient vector fields associated with some smooth functions $I_1,\dots, I_k, D_1, \dots, D_p \in \mathcal{C}^{\infty}(U,\mathbb{R})$, i.e., $X_1 =\nabla_{g}I_1$, $\dots$, $X_k = \nabla_{g}I_k$, $Y_1 = \nabla_{g}D_1$, $\dots, Y_p = \nabla_{g}D_p$, then we obtain a set of local generators for the smooth affine distribution of smooth vector fields which conserve the quantities $I_1, \dots, I_k$ and dissipate the scalar quantities $D_1, \dots, D_p$ with prescribed rates $h_1, \dots, h_p$.
\end{abstract}

\medskip

\textbf{AMS 2000}: 58A30; 53B21; 37C10.

\textbf{Keywords}: affine distributions; Riemannian manifolds; contravariant Grassmann algebra; conservative/dissipative dynamics.

\section{Introduction}
\label{section:one}

The theory of distributions, in the broadest sense, it is perhaps one of the most influential tool from differential geometry, according to its usefulness in a large assortment of scientific domains, e.g., geometric control theory, differential equations, sub-Riemannian geometry, dynamical systems, nonholonomic mechanics (for details see, e.g., \cite{Boothby}, \cite{stefan}, \cite{susman}, \cite{Calin}, \cite{Bloch}, \cite{ratiurazvan}). 

The main protagonist of this paper is a special class of distributions, namely the so called smooth affine distributions on Riemannian manifolds. More precisely, the main purpose of this work is to provide explicitly (and coordinate free) a set of local generators for a given smooth affine distribution on a finite dimensional smooth Riemannian manifold. Moreover, by applying this result to some special classes of affine distributions, we generalize some results from \cite{BC} related to smooth affine distributions associated to dissipative dynamical systems.

More exactly, in the second section we provide a coordinate free formulation for the intersection of a finite number of hyperplanes of a finite dimensional inner product space (Euclidean vector space).

The third section is dedicated to the local study of smooth affine distributions on Riemannian manifolds. Since the main applications of our results are supposed to improve the study of conservative/dissipative dynamical systems, for practical reasons, we will focus on the local study of smooth affine distributions. More precisely, the main purpose of this section is to provide explicitly (and coordinate free) a set of local generators for a smooth affine distribution given by those vector fields $X\in\mathfrak{X}(U)$ defined eventually on an open subset $U\subseteq M$ of a smooth Riemannian manifold $(M,g)$, that verifies the relations $g(X,X_1)=\dots =g(X,X_k)=0$, $g(X,Y_1)=h_1, \dots, g(X,Y_p)=h_p$, where $X_1,\dots, X_k, Y_1, \dots, Y_p \in\mathfrak{X}(U)$, and respectively $h_1,\dots, h_p \in \mathcal{C}^{\infty}(U,\mathbb{R})$ are a-priori given quantities. The analysis of mixed homogeneous and nonhomogeneous relations is deliberate (even if one can recover the homogeneous part by simply annihilating the nonhomogeneous one) because of the clarity of formulas for the local generators. Moreover, these two classes of relations have completely different meaning in dynamical setting, as can be seen in the next section.

The aim of the last section is to apply the results obtained in the previous section to some special classes of smooth affine distributions naturally associated to dynamical systems, and also to provide a unified presentation of conservative and dissipative dynamical systems. More exactly, in the case when $X_1,\dots, X_k, Y_1, \dots, Y_p$ are gradient vector fields on a smooth Riemannian manifold $(M,g)$ associated with some smooth functions $I_1,\dots, I_k, D_1, \dots, D_p \in \mathcal{C}^{\infty}(U,\mathbb{R})$, i.e., $X_1 =\nabla_{g}I_1$, $\dots$, $X_k = \nabla_{g}I_k$, $Y_1 = \nabla_{g}D_1$, $\dots, Y_p = \nabla_{g}D_p$, then we obtain a set of local generators for the smooth affine distribution of vector fields which conserve the quantities $I_1, \dots, I_k$ and dissipate the scalar quantities $D_1, \dots, D_p$ with prescribed rates $h_1, \dots, h_p$. As a consequence one obtain a generalization for $p>1$ of a result from \cite{BC} given for $p=1$. Note that for $p>0$ the main result provides a local characterization of dissipative dynamical systems. Some other dynamically relevant cases are obtained, e.g., for $p=0$ one obtain a local characterization of conservative dynamical systems; for $p=0$ and $k=\dim{M}-1$ one obtain a local characterization of completely integrable dynamical systems.

\section{A coordinate free formulation of hyperplanes intersection}

In this section we obtain a coordinate free formulation for the linear variety determined by the intersection of a finite number of hyperplanes of a finite dimensional inner product space.

In order to do that, let us recall that given an $n$-dimensional inner product space $(E,\langle\cdot,\cdot\rangle)$ over a field $\mathbb{K}$ of characteristic zero, then for any $p\in\{1,\dots, n\}$, the $p$-th exterior power of the vector space $E$, $\Lambda^{p}E$, inherits an inner product, $\langle\cdot,\cdot\rangle_{p}$, defined on pairs of decomposable $p$-vectors by
\begin{equation}\label{ipmv} 
\langle v_1 \wedge \dots \wedge v_p,  w_1 \wedge \dots \wedge w_p \rangle_{p} := \det[\langle v_i,w_j \rangle_{1\leq i,j\leq p}],
\end{equation}
for any $v_1 \wedge \dots \wedge v_p,  w_1 \wedge \dots \wedge w_p \in \Lambda^{p}E$, and extended by bilinearity to the whole vector space $\Lambda^{p}E$. Note that $(\Lambda^{1}E, \langle\cdot,\cdot\rangle_{1})=(E,\langle\cdot,\cdot\rangle)$, and by convention $\Lambda^{0}E=\mathbb{K}$.

As usual, one denote by $\| \cdot \|_{p} =\sqrt{\langle \cdot, \cdot \rangle_{p}}$, the norm induced by the inner product $\langle\cdot,\cdot\rangle_{p}$. Note that in the case of the inner product $\langle\cdot,\cdot\rangle_{p}$, $$\| v_1 \wedge \dots \wedge v_p \|_{p}= \sqrt{\det(G(v_1,\dots, v_p))},$$ where $G(v_1,\dots, v_p)=[\langle v_i,v_j \rangle_{1\leq i,j\leq p}]$ is the Gram matrix associated to the ordered set of vectors $\{v_1, \dots, v_p\}\subset E$.

Recall that any orthonormal basis of $E$, $\{e_1,\dots,e_n\}$, generates an orthonormal basis of $\Lambda^{p}E$, $$\{e_{i_1}\wedge\dots\wedge e_{i_p} \mid 1\leq i_1 \leq \dots \leq i_p \leq n\}.$$

Consequently, for any $p\in\{0,\dots,n\}$, we have $$\dim_{\mathbb{K}}\Lambda^{p}E = \binom{n}{p}=\binom{n}{n-p}=\dim_{\mathbb{K}}\Lambda^{n-p}E,$$ and hence $\Lambda^{p}E\cong \Lambda^{n-p}E$. A natural isomorphism between these vector spaces is given by the Hodge star operator. 

In order to remind the definition of the Hodge star operator let us fix an orthonormal basis of the vector space $E$, say $\{e_1,\dots, e_n\}$, and the corresponding  basis unit vector $\mu = e_1 \wedge\dots\wedge e_n$ for the vector space $\Lambda^{n}E$. Note that since the basis $\{e_1,\dots, e_n\}$ is orthonormal, we get that $\|\mu\|_{n}=1$. The volume element $\mu$ is hence unique up to a sign and defines an orientation of $E$.

Fixing $p\in\{1,\dots, n\}$ and $\nu \in \Lambda^{p}E$, we get that the map $$ \omega\in\Lambda^{n-p}E \mapsto \nu\wedge\omega\in\Lambda^{n}E,$$ is linear, and hence there exists a unique linear functional $\alpha_{\nu}\in\left(\Lambda^{n-p}E\right)^{\star}$ such that $$\nu\wedge\omega=\alpha_{\nu}(\omega)\mu.$$ Since $\Lambda^{n-p}E$ is an inner product space, due to Riesz representation of linear functionals, we have the existence of a unique element of $\Lambda^{n-p}E$, denoted $\star \nu$, such that for any $\omega\in \Lambda^{n-p}E$,  $$\alpha_{\nu}(\omega)=\langle \star\nu,\omega\rangle_{n-p}.$$ 

Hence, for any $\nu\in\Lambda^{p}E$ and $\omega\in\Lambda^{n-p}E$, we have
\begin{equation}\label{Hs}
\langle \star\nu,\omega\rangle_{n-p} \quad \mu=\nu\wedge\omega.
\end{equation}

Summarizing, the linear operator $\nu \in \Lambda^{p}E \mapsto \star \nu \in \Lambda^{n-p}E$, is by construction an isomorphism of vector spaces, and is called the Hodge star operator. Moreover, a direct consequence of the property \eqref{Hs} and of the fact that $(\star\circ\star) (\nu) = (-1)^{p(n-p)}\nu$, for any $\nu\in\Lambda^{p}E$, is that $$\langle \star \nu_1 , \star \nu_2 \rangle_{n-p}=\langle \nu_1 , \nu_2 \rangle_{p},$$
for any $\nu_1 ,\nu_2 \in\Lambda^{p}E$.

For more details regarding the Hodge star operator and other properties of $p$-th exterior powers of vector spaces, see e.g., \cite{darling}.

Before stating the main result of this section, let us give an auxiliary result that provides a coordinate free formulation for the intersection of $k$ linear hyperplanes with prescribed normal directions, in a finite dimensional inner product space.

\begin{proposition}\label{aux}
Let $(E,\langle\cdot,\cdot\rangle)$ be an $n$-dimensional inner product space over a field $\mathbb{K}$ of characteristic zero, and let $\{v_1,\dots, v_k\}\subset E$ be a set of linearly independent vectors ($k\in \mathbb{N}$, $0<k<n-1$). 
Then the solutions $u\in E$ of the system
\begin{equation*}
\langle u, v_1\rangle = \dots = \langle u, v_k\rangle = 0,
\end{equation*}
are the elements of the $(n-k)$-dimensional vector subspace 
$$
E[v_1,\dots,v_k]:=\operatorname{span}_{\mathbb{K}}\left\{\star\left( \bigwedge_{i=1, i\neq a}^{n-k} \omega_i \wedge \bigwedge_{l=1}^{k} v_l \right): a\in\{1,\dots, n-k\}\right\},
$$
where $$\{\omega_1, \dots, \omega_{n-k}\} \subset E $$ is a set of linearly independent vectors such that $\{v_1,\dots, v_k, \omega_1, \dots, \omega_{n-k}\}$ forms a basis of $E$.
\end{proposition}
\begin{proof}
Recall first that the intersection of $k$ linear hyperplanes with linearly independent normal directions, of an $n$ dimensional vector space is a vector subspace of dimension $n-k$. Hence, in order to obtain the conclusion will be enough to find a set of $n-k$ linearly independent solutions of the system $\langle u, v_1\rangle = \dots = \langle u, v_k\rangle = 0$ which also belong to $E[v_1,\dots,v_k]$. Note that the vector space $E[v_1,\dots,v_k]$ is independent of the choice of the set of linear independent vectors $\{\omega_1, \dots, \omega_{n-k}\} \subset E$ such that $\{v_1,\dots, v_k, \omega_1, \dots, \omega_{n-k}\}$ forms a basis of $E$. Indeed, let us fix another set of independent vectors $\{w_1, \dots, w_{n-k}\} \subset E$ such that $\{v_1,\dots, v_k, w_1, \dots, w_{n-k}\}$ forms a basis of $E$. Then, writing for each $i\in\{1,\dots, n-k\}$, the vector $w_i$ with respect to the basis $\{v_1,\dots, v_k, \omega_1, \dots, \omega_{n-k}\}$, namely, $$w_i = \sum_{l=1}^{k} \lambda_{il} v_l + \sum_{b=1}^{n-k} \nu_{ib} \omega_b,$$ (where $\lambda_{il}, \nu_{ib}\in \mathbb{K}$ are such that for each $i\in\{1,\dots, n-k\}$ fixed, the set $S_i:=\{\nu_{ib}: b\in\{1,\dots,n-k\}\}$ contains at least one non-zero element, and moreover for each $i,j\in\{1,\dots,n-k\}$, $i\neq j$, the vectors $\sum_{b=1}^{n-k} \nu_{ib} \omega_b$ and $\sum_{b=1}^{n-k} \nu_{jb} \omega_b$ are linearly independent, since $w$'s and $v$'s are linearly independent), and using the multilinear properties of the wedge product and the linearity of the Hodge star operator we get 
\begin{align*}
\operatorname{span}_{\mathbb{K}}\left\{\star\left( \bigwedge_{i=1, i\neq a}^{n-k} w_i \wedge \bigwedge_{l=1}^{k} v_l \right): a\in\{1,\dots, n-k\}\right\} \\
=\operatorname{span}_{\mathbb{K}}\left\{\star\left[ \bigwedge_{i=1, i\neq a}^{n-k} \left(\sum_{l=1}^{k} \lambda_{il} v_l + \sum_{b=1}^{n-k} \nu_{ib} \omega_b\right) \wedge \bigwedge_{l=1}^{k} v_l \right]: a\in\{1,\dots, n-k\}\right\} \\
=\operatorname{span}_{\mathbb{K}}\left\{\star\left( \bigwedge_{i=1, i\neq a}^{n-k} \omega_i \wedge \bigwedge_{l=1}^{k} v_l \right): a\in\{1,\dots, n-k\}\right\}.
\end{align*}

Let us now show that the vectors $$u_a :=\star\left( \bigwedge_{i=1, i\neq a}^{n-k} \omega_i \wedge \bigwedge_{l=1}^{k} v_l \right),$$ for $a\in \{1,\dots, n-k\}$, are linearly independent solutions of the system $$\langle u, v_1\rangle = \dots = \langle u, v_k\rangle = 0.$$ 

For proving the linear independence of the vectors $\{u_1,\dots,u_{n-k}\}$, since the Hodge star operator is an isomorphism, it is enough to show the linear independence of the $(n-1)$-vectors $\{\star u_1,\dots,\star u_{n-k}\}$. In order to do that, let us consider $a_1,\dots,a_{n-k}\in\mathbb{K}$ such that
\begin{equation*}
\sum_{s=1}^{n-k} a_s \left((-1)^{n-1}\bigwedge_{i=1, i\neq s}^{n-k} \omega_i \wedge \bigwedge_{l=1}^{k} v_l \right) =0.
\end{equation*}
Wedging both members of the above equality with $\omega_r$ for a fixed $r\in\{1,\dots, n-k\}$, one obtain $$a_r \bigwedge_{i=1}^{n-k}\omega_i \wedge\bigwedge_{l=1}^{k}v_l =0,$$ and consequently $a_r =0$, since $\bigwedge_{i=1}^{n-k}\omega_i \wedge\bigwedge_{l=1}^{k}v_l\neq 0$, because $\{v_1,\dots, v_k, \omega_1, \dots, \omega_{n-k}\}$ are linearly independent. Repeating the argument for each $r\in\{1,\dots, n-k\}$ one obtain $a_1 = \dots = a_{n-k}=0$, and hence the linear independence of the $(n-1)$-vectors $\{\star u_1,\dots,\star u_{n-k}\}$.

The last step of the proof is to verify that for each $a\in\{1,\dots,n-k\}$, $u_a$ is a solution of the system $\langle u, v_1\rangle = \dots = \langle u, v_k\rangle = 0$.
Note that for arbitrary $a\in\{1,\dots,n-k\}$ and $i\in\{1,\dots,k\}$, by using the formula \eqref{Hs} we obtain
\begin{align*}
\langle u_a,v_i\rangle \mu &=\langle \star\left( \bigwedge_{i=1, i\neq a}^{n-k} \omega_i \wedge \bigwedge_{l=1}^{k} v_l \right),v_i\rangle \mu \\
&=\left(\bigwedge_{i=1, i\neq a}^{n-k} \omega_i \wedge \bigwedge_{l=1}^{k} v_l \right)\wedge v_i\\
&=\bigwedge_{i=1, i\neq a}^{n-k} \omega_i \wedge\left( \bigwedge_{l=1}^{k} v_l \wedge v_i \right)\\
&=0.
\end{align*}
Since from the definition of the Hodge star operator, $\mu \neq 0$, we obtain $\langle u_a,v_i\rangle =0$.
\end{proof}

Let us point out that the Proposition \eqref{aux} remains valid also for $k\in\{n-1,n\}$, the only difference from the case $0<k<n-1$ being the inconsistency of the notations, which in these limit cases may lead to confusions. Hence, for these limit cases we prefer to state separately the conclusion of the Proposition \eqref{aux}.

\begin{remark}\label{re0}
\begin{itemize}
\item For $k=n$, the Proposition \eqref{aux} becomes trivial since the only solution of the system $\langle u, v_1\rangle = \dots = \langle u, v_n\rangle = 0$ is $u=0$. 

\item For $k=n-1$, the conclusion of Proposition \eqref{aux} becomes as follows:

The solutions $u\in E$ of the system
\begin{equation*}
\langle u, v_1\rangle = \dots = \langle u, v_{n-1}\rangle = 0,
\end{equation*}
are the elements of the $1$-dimensional vector subspace 
$$
E[v_1,\dots,v_{n-1}]:=\operatorname{span}_{\mathbb{K}}\left\{\star\left( \bigwedge_{l=1}^{n-1} v_l \right)\right\}.
$$
\end{itemize}
\end{remark}

Let us give now the main result of this section, which provides a coordinate free formulation of the linear variety described by the intersection of $k$ linear hyperplanes and respectively $p$ affine hyperplanes of an $n$ - dimensional inner product space $(E,\langle\cdot,\cdot\rangle)$.  

\begin{theorem}\label{MT}
Let $(E,\langle\cdot,\cdot\rangle)$ be an $n$-dimensional inner product space over a field $\mathbb{K}$ of characteristic zero. Let $k,p\in \mathbb{N}$, $k>0$, $p>1$, $k+p < n-1$, $\lambda_1, \dots, \lambda_p \in \mathbb{K}\setminus \{0\}$ be given, and let $\{v_1,\dots, v_k, w_1, \dots, w_p\}\subset E$ be a set of linearly independent vectors. 

Then the solutions $u\in E$ of the system
\begin{equation}\label{gsis}
\left\{\begin{array}{l}
\langle u, v_1\rangle = \dots = \langle u, v_k\rangle = 0,\\
\langle u, w_1\rangle = \lambda_1, \dots, \langle u, w_p\rangle = \lambda_p,\\
\end{array}\right.
\end{equation}
are given by $u = u_{0} + u_{\perp}$, where
\begin{equation*}
u_{0}=\left\| \bigwedge_{i=1}^{p} w_i\wedge\bigwedge_{j=1}^{k} v_j \right\|_{k+p}^{-2}\cdot\sum_{i=1}^{p}(-1)^{n-i}\lambda_i \Theta_i, 
\end{equation*}
with $$\Theta_i = \star\left[ \bigwedge_{j=1, j\neq i}^{p} w_j \wedge \bigwedge_{l=1}^{k} v_l  \wedge\star\left(\bigwedge_{j=1}^{p} w_j\wedge\bigwedge_{l=1}^{k} v_l\right)\right],$$
and $$u_{\perp}\in \operatorname{span}_{\mathbb{K}}\left\{\star\left( \bigwedge_{i=1, i\neq a}^{n-(k+p)} \omega_i \wedge \bigwedge_{j=1}^{p} w_j\wedge\bigwedge_{l=1}^{k} v_l \right): a\in\{1,\dots, n-(k+p)\}\right\},$$ where $$\{\omega_1, \dots, \omega_{n-(k+p)}\} \subset E $$ is a set of linearly independent vectors such that $\{v_1,\dots, v_k, w_1, \dots, w_p, \omega_1, \dots, \omega_{n-(k+p)}\}$ forms a basis of $E$.
\end{theorem}
\begin{proof}
Recall first that the intersection of $k+p$ (affine) hyperplanes with linearly independent normal directions, of an $n$ dimensional vector space is a linear variety of dimension $n-(k+p)$, whose direction is given by the intersection of the associated linear hyperplanes, namely the solutions set of the system
\begin{equation*}
\left\{\begin{array}{l}
\langle u, v_1\rangle = \dots = \langle u, v_k\rangle = 0,\\
\langle u, w_1\rangle = \dots, \langle u, w_p\rangle = 0,\\
\end{array}\right.
\end{equation*}
which by Proposition \eqref{aux}, is given by the vector subspace
\begin{align*}
E[w_1,\dots, & w_p,v_1,\dots,v_k] = \\
&=\operatorname{span}_{\mathbb{K}}\left\{\star\left( \bigwedge_{i=1, i\neq a}^{n-(k+p)} \omega_i \wedge \bigwedge_{j=1}^{p} w_j\wedge\bigwedge_{l=1}^{k} v_l \right) : a\in\{1,\dots, n-(k+p)\}\right\},
\end{align*}
where $$\{\omega_1, \dots, \omega_{n-(k+p)}\} \subset E $$ is a set of linearly independent vectors such that $\{v_1,\dots, v_k, w_1, \dots, w_p, \omega_1, \dots, \omega_{n-(k+p)}\}$ forms a basis of $E$.

Note that the solutions set of the system \eqref{gsis} is the linear variety 
\begin{equation*}
u_0 + E[w_1,\dots,  w_p,v_1,\dots , v_k] \subset {E},
\end{equation*}
where $u_0 \in E$ is a particular solution of the system \eqref{gsis}. 

Hence, in order to complete the proof of the theorem it is enough to check that 
\begin{equation*}
u_{0}=\left\| \bigwedge_{i=1}^{p} w_i\wedge\bigwedge_{j=1}^{k} v_j \right\|_{k+p}^{-2}\cdot\sum_{i=1}^{p}(-1)^{n-i}\lambda_i \Theta_i, 
\end{equation*}
with $$\Theta_i = \star\left[ \bigwedge_{j=1, j\neq i}^{p} w_j \wedge \bigwedge_{l=1}^{k} v_l  \wedge\star\left(\bigwedge_{j=1}^{p} w_j\wedge\bigwedge_{l=1}^{k} v_l\right)\right],$$
verifies the system \eqref{gsis}.

Let $v_b\in\{v_1,\dots,v_k\}$ be arbitrary fixed. Using the linearity of the inner product we obtain
\begin{equation*}
\langle u_0, v_b\rangle =\left\| \bigwedge_{i=1}^{p} w_i\wedge\bigwedge_{j=1}^{k} v_j \right\|_{k+p}^{-2}\cdot\sum_{i=1}^{p}(-1)^{n-i}\lambda_i \langle \Theta_i , v_b\rangle.
\end{equation*}
In order to show that $\langle u_0, v_b\rangle=0$, it is enough to prove that $\langle \Theta_i , v_b\rangle =0$, for every $i\in\{1,\dots,p\}$. Indeed, using the properties of the Hodge star operator we obtain:
\begin{align*}
\langle \Theta_i, v_b\rangle \mu &=\langle \star\left[ \bigwedge_{j=1, j\neq i}^{p} w_j \wedge \bigwedge_{l=1}^{k} v_l  \wedge\star\left(\bigwedge_{j=1}^{p} w_j\wedge\bigwedge_{l=1}^{k} v_l\right)\right], v_b\rangle \mu\\
&=\left[ \bigwedge_{j=1, j\neq i}^{p} w_j \wedge \bigwedge_{l=1}^{k} v_l  \wedge\star\left(\bigwedge_{j=1}^{p} w_j\wedge\bigwedge_{l=1}^{k} v_l\right)\right]\wedge v_b\\
&=(-1)^{k[n-(p+k)]}\bigwedge_{j=1, j\neq i}^{p} w_j \wedge \star\left(\bigwedge_{j=1}^{p} w_j\wedge\bigwedge_{l=1}^{k} v_l\right)\wedge \left(\bigwedge_{l=1}^{k} v_l  \wedge v_b \right)\\
&=0,
\end{align*}
and hence $\langle \Theta_i, v_b\rangle=0$.

Last step of the proof is to verify that $\langle u_0, w_d\rangle = \lambda_d$, for every $d\in\{1,\dots, p\}$. In order to do that, let $d\in\{1,\dots, p\}$ be fixed.
Using the linearity of the inner product we obtain
\begin{equation}\label{eer}
\langle u_0, w_d\rangle =\left\| \bigwedge_{i=1}^{p} w_i\wedge\bigwedge_{j=1}^{k} v_j \right\|_{k+p}^{-2}\cdot\sum_{i=1}^{p}(-1)^{n-i}\lambda_i \langle \Theta_i , w_d\rangle. 
\end{equation}
Let us first evaluate the terms $\langle \Theta_i , w_d\rangle$, for $i\in\{1,\dots,p\}$.
Using the properties of the Hodge star operator we obtain:
\begin{align*}
\langle \Theta_i, w_d\rangle \mu &=\langle \star\left[ \bigwedge_{j=1, j\neq i}^{p} w_j \wedge \bigwedge_{l=1}^{k} v_l  \wedge\star\left(\bigwedge_{j=1}^{p} w_j\wedge\bigwedge_{l=1}^{k} v_l\right)\right], w_d\rangle \mu\\
&=\left[ \bigwedge_{j=1, j\neq i}^{p} w_j \wedge \bigwedge_{l=1}^{k} v_l  \wedge\star\left(\bigwedge_{j=1}^{p} w_j\wedge\bigwedge_{l=1}^{k} v_l\right)\right]\wedge w_d\\
&=(-1)^{n-p}\left(\bigwedge_{j=1, j\neq i}^{p} w_j \wedge w_d \right)\wedge \bigwedge_{l=1}^{k} v_l \wedge\star\left(\bigwedge_{j=1}^{p} w_j\wedge\bigwedge_{l=1}^{k} v_l\right)\\
\end{align*}
\begin{align*}
&=\left\{\begin{array}{l}
0, \quad \text{if} \quad d\neq i\\
(-1)^{n-p}(-1)^{p-i}\left(\bigwedge_{j=1}^{p} w_j\wedge\bigwedge_{l=1}^{k} v_l\right)\wedge\star\left(\bigwedge_{j=1}^{p} w_j\wedge\bigwedge_{l=1}^{k} v_l\right),\quad  \text{if} \quad d= i\\
\end{array}\right.\\
&=\left\{\begin{array}{l}
0, \quad \text{if} \quad d\neq i\\
(-1)^{n-i}\left\| \star\left(\bigwedge_{j=1}^{p} w_j\wedge\bigwedge_{l=1}^{k} v_l \right)\right\|_{n-(k+p)}^{2}\mu, \quad \text{if} \quad d= i\\
\end{array}\right.\\
&=\left\{\begin{array}{l}
0, \quad \text{if} \quad d\neq i\\
(-1)^{n-i}\left\|\bigwedge_{j=1}^{p} w_j\wedge\bigwedge_{l=1}^{k} v_l \right\|_{k+p}^{2}\mu, \quad \text{if} \quad d= i,\\
\end{array}\right.\\
\end{align*}
and hence $\langle \Theta_i, w_d\rangle =(-1)^{n-i}\left\|\bigwedge_{j=1}^{p} w_j\wedge\bigwedge_{l=1}^{k} v_l \right\|_{k+p}^{2} \delta_{id}$, where $\delta$ stands for Kronecker's delta, namely $\delta_{id}=0$ for $i\neq d$, and respectively $\delta_{id}=1$ for $i=d$.

Replacing the terms $\langle \Theta_i, w_d\rangle=(-1)^{n-i}\left\|\bigwedge_{j=1}^{p} w_j\wedge\bigwedge_{l=1}^{k} v_l \right\|_{k+p}^{2} \delta_{id}$ in the equality \eqref{eer} we obtain 
\begin{align*}
\langle u_0, w_d\rangle &=\left\| \bigwedge_{i=1}^{p} w_i\wedge\bigwedge_{j=1}^{k} v_j \right\|_{k+p}^{-2}\cdot\sum_{i=1}^{p}(-1)^{n-i}\lambda_i \langle \Theta_i , w_d\rangle\\
&=\left\| \bigwedge_{i=1}^{p} w_i\wedge\bigwedge_{j=1}^{k} v_j \right\|_{k+p}^{-2}\cdot\sum_{i=1}^{p}(-1)^{n-i}\lambda_i (-1)^{n-i}\left\|\bigwedge_{j=1}^{p} w_j\wedge\bigwedge_{l=1}^{k} v_l \right\|_{k+p}^{2} \delta_{id}\\
&=\sum_{i=1}^{p}\lambda_i \delta_{id}\\
&=\lambda_d,
\end{align*}
and hence $\langle u_0, w_d\rangle=\lambda_d$.
\end{proof}

\begin{remark}
If one adopt the notation $\mathcal{H}_{(v_i;0)}:=\{u\in E \mid \langle u, v_i\rangle =0\}$, $i\in\{1,\dots, k\}$, and respectively $\mathcal{H}_{(w_j;\lambda_j)}:=\{u\in E \mid \langle u, w_j\rangle =\lambda_j\}$, $j\in\{1,\dots, p\}$, then the intersection of the above defined linear and respectively affine hyperplanes, is the linear variety
\begin{equation}\label{HypInt}
\bigcap_{i=1}^{k}\mathcal{H}_{(v_i;0)}\cap\bigcap_{j=1}^{p}\mathcal{H}_{(w_j;\lambda_j)}= u_0 + E[w_1,\dots, w_p,v_1,\dots,v_k],
\end{equation}
where 

\begin{equation*}
u_{0}=\left\| \bigwedge_{i=1}^{p} w_i\wedge\bigwedge_{j=1}^{k} v_j \right\|_{k+p}^{-2}\cdot\sum_{i=1}^{p}(-1)^{n-i}\lambda_i \Theta_i, 
\end{equation*}
with $$\Theta_i = \star\left[ \bigwedge_{j=1, j\neq i}^{p} w_j \wedge \bigwedge_{l=1}^{k} v_l  \wedge\star\left(\bigwedge_{j=1}^{p} w_j\wedge\bigwedge_{l=1}^{k} v_l\right)\right],$$
and 
\begin{align*}
E[w_1,\dots, & w_p,v_1,\dots,v_k] = \bigcap_{i=1}^{k}\mathcal{H}_{(v_i;0)}\cap\bigcap_{j=1}^{p}\mathcal{H}_{(w_j;0)}\\
&=\operatorname{span}_{\mathbb{K}}\left\{\star\left( \bigwedge_{i=1, i\neq a}^{n-(k+p)} \omega_i \wedge \bigwedge_{j=1}^{p} w_j\wedge\bigwedge_{l=1}^{k} v_l \right) : a\in\{1,\dots, n-(k+p)\}\right\},
\end{align*}
where $$\{\omega_1, \dots, \omega_{n-(k+p)}\} \subset E $$ is a set of linearly independent vectors such that $\{v_1,\dots, v_k, w_1, \dots, w_p, \omega_1, \dots, \omega_{n-(k+p)}\}$ forms a basis of $E$.
\end{remark}

Let us point out that the Theorem \eqref{MT} remains valid also for the limit cases $p\in\{0,1\}$, $k=0$, $k+p\in\{n-1,n\}$, the only difference from the general case being the inconsistency of the notations, which in these limit cases may lead to confusions. Hence, for these limit cases we prefer to state separately the conclusion of the Theorem \eqref{MT}.

\begin{remark}\label{re2}
\begin{itemize}
\item For $p=0$, the Theorem \eqref{MT} reduces to Proposition \eqref{aux}. 

\item For $p=1$, the conclusion of the Theorem \eqref{MT} becomes as follows:

The solutions $u\in E$ of the system
\begin{equation*}
\left\{\begin{array}{l}
\langle u, v_1\rangle = \dots = \langle u, v_k\rangle = 0,\\
\langle u, w_1\rangle = \lambda_1,\\
\end{array}\right.
\end{equation*}
are given by $u = u_{0} + u_{\perp}$, where
\begin{equation*}
u_{0}=\left\| w_1\wedge\bigwedge_{j=1}^{k} v_j \right\|_{k+1}^{-2}\cdot(-1)^{n-1}\lambda_1 \Theta_1, 
\end{equation*}
with $$\Theta_1 = \star\left[ \bigwedge_{l=1}^{k} v_l  \wedge\star\left(w_1\wedge\bigwedge_{l=1}^{k} v_l\right)\right],$$
and $$u_{\perp}\in \operatorname{span}_{\mathbb{K}}\left\{\star\left( \bigwedge_{i=1, i\neq a}^{n-(k+1)} \omega_i \wedge w_1 \wedge\bigwedge_{l=1}^{k} v_l \right): a\in\{1,\dots, n-(k+1)\}\right\},$$ where $$\{\omega_1, \dots, \omega_{n-(k+1)}\} \subset E $$ is a set of linearly independent vectors such that $\{v_1,\dots, v_k, w_1, \omega_1, \dots, \omega_{n-(k+1)}\}$ forms a basis of $E$.
\end{itemize}
\end{remark}

\begin{remark}\label{re3}
For $k=0$, the conclusion of the Theorem \eqref{MT} becomes as follows:

The solutions $u\in E$ of the system
\begin{equation*}
\langle u, w_1\rangle = \lambda_1, \dots, \langle u, w_p\rangle = \lambda_p,\\
\end{equation*}
are given by $u = u_{0} + u_{\perp}$, where
\begin{equation*}
u_{0}=\left\| \bigwedge_{i=1}^{p} w_i \right\|_{p}^{-2}\cdot\sum_{i=1}^{p}(-1)^{n-i}\lambda_i \Theta_i, 
\end{equation*}
with $$\Theta_i = \star\left[ \bigwedge_{j=1, j\neq i}^{p} w_j \wedge\star\left(\bigwedge_{j=1}^{p} w_j\right)\right],$$
and $$u_{\perp}\in \operatorname{span}_{\mathbb{K}}\left\{\star\left( \bigwedge_{i=1, i\neq a}^{n-p} \omega_i \wedge \bigwedge_{j=1}^{p} w_j \right): a\in\{1,\dots, n-p\}\right\},$$ where $$\{\omega_1, \dots, \omega_{n-p}\} \subset E $$ is a set of linearly independent vectors such that $\{w_1, \dots, w_p, \omega_1, \dots, \omega_{n-p}\}$ forms a basis of $E$.
\end{remark}

\begin{remark}\label{re1}
In the case when $k+p=n-1$, the above theorem conclusions still hold true, the only difference being the fact that the direction of the linear variety is one-dimensional and can be expressed only in terms of the vectors $v_1,\dots, v_k, w_1, \dots, w_p$:
$$
u_{\perp}\in \operatorname{span}_{\mathbb{K}}\left\{\star\left( \bigwedge_{j=1}^{p} w_j\wedge\bigwedge_{l=1}^{k} v_l \right)\right\}
$$ 
\end{remark}

\begin{remark}\label{re01}
For $k+p=n$, the conclusion of the Theorem \eqref{MT} becomes as follows:

The system
\begin{equation*}
\left\{\begin{array}{l}
\langle u, v_1\rangle = \dots = \langle u, v_k\rangle = 0,\\
\langle u, w_1\rangle = \lambda_1, \dots, \langle u, w_p\rangle = \lambda_p,\\
\end{array}\right.
\end{equation*}
has a unique solution which is given by $u = u_{0}$, where
\begin{equation*}
u_{0}=\left\| \bigwedge_{i=1}^{p} w_i\wedge\bigwedge_{j=1}^{k} v_j \right\|_{n}^{-2}\cdot\sum_{i=1}^{p}(-1)^{n-i}\lambda_i \Theta_i, 
\end{equation*}
with $$\Theta_i = \star\left[ \bigwedge_{j=1, j\neq i}^{p} w_j \wedge \bigwedge_{l=1}^{k} v_l  \wedge\star\left(\bigwedge_{j=1}^{p} w_j\wedge\bigwedge_{l=1}^{k} v_l\right)\right].$$
\end{remark}

\section{Local generators of affine distributions on Riemannian manifolds}

The purpose of this section is to translate on smooth Riemannian manifolds the results given in the previous section. This approach follows naturally, and has direct applications to dynamical systems. As we will see in the next section, the results presented here will provide an explicit characterization of conservative and also dissipative dynamical systems.

Let us start by recalling the Riemannian version of the main protagonists of previous section. In order to do that, let $(M,g)$ be a smooth $n$-dimensional Riemannian manifold. Recall that the Riemannian metric $g$ induces an inner product space structure on each tangent space $(T_x M, g(x) =:\langle\cdot,\cdot\rangle_x )$, the assignment $x\mapsto g(x)$ depending smoothly on the base point $x\in M$. 

Recall also that for any $p\in\{1,\dots, n\}$, $\Lambda^{p}TM$, the $p$-th order exterior power of the tangent bundle $TM$, is a vector bundle over $M$ whose fiber on each point $x\in M$ is the vector space $\Lambda^{p}T_x M$. The smooth sections of $\Lambda^{p}TM$, $\Gamma^{\infty}(\Lambda^{p}TM)=:\mathfrak{X}^{p}(M)$, are called smooth $p$-vector fields ($\mathfrak{X}^{1}(M)=\mathfrak{X}(M)$). Note that since $(M,g)$ is a Riemannian manifold, one have a natural assignment, $x\mapsto g_p (x)$, depending smoothly on $x\in M$, where $(\Lambda^{p}T_x M, g_p(x))$ is an inner product space together with the natural inner product $g_p (x)$ induced by the inner product $g(x)$ from $T_x M$, through the formula \eqref{ipmv}. 

More precisely, the map $g_p:\Gamma^{\infty}(\Lambda^{p}TM)\times \Gamma^{\infty}(\Lambda^{p}TM)\rightarrow \mathcal{C}^{\infty}(M,\mathbb{R})$ is given by $(g_p (X_1,X_2))(x)=g_p (x)(X_1 (x), X_2 (x))$, for any $X_1, X_2 \in \Gamma^{\infty}(\Lambda^{p}TM)=\mathfrak{X}^{p}(M)$, and any $x\in M$. Recall that $\|X\|_{p}=:\sqrt{g_p(X,X)}$, for any $X\in \mathfrak{X}^{p}(M)$.

Analogously, for every $p\in\{0,\dots, n\}$, the smooth family of Hodge star operators, $\star_x : \Lambda^{p}T_x M \rightarrow \Lambda^{n-p}T_x M$, induces a base point preserving operator $\star : \Gamma^{\infty}(\Lambda^{p}TM)\rightarrow \Gamma^{\infty}(\Lambda^{n-p}TM)$, namely the Hodge star operator on multivector fields, $\star : \mathfrak{X}^{p}(M)\rightarrow \mathfrak{X}^{n-p}(M)$. For more details regarding Riemannian manifolds see, e.g., \cite{jost}.

Before stating the main result of this section, let us recall that a moving frame on an open subset $\mathcal{O}$ of an $n$-dimensional smooth manifold $M$, consists of a set of $n$ locally defined smooth vector fields, $\{X_1,\dots, X_n\}\subset \mathfrak{X}(\mathcal{O})$, such that $\operatorname{span}_{\mathbb{R}}\{X_1(x),\dots, X_n(x)\}=T_x M$, for each $x\in \mathcal{O}$. Note that moving frames always exist locally, in some open neighborhood around any given point of the manifold, but in general they are not globally defined. The existence of globally defined moving frames is equivalent to the triviality of the tangent bundle, $TM$. The manifolds with globally defined moving frames are called parallelizable, e.g., Lie groups. In the case when $M$ is a Riemannian manifold $(M,g)$, to any moving frame one can associate an orthogonal (orthonormal) moving frame, that is, a frame consisting of orthogonal (unit) vectors at each point.

Let us state now the main result of this sections, which is the equivalent of the Theorem \eqref{MT}, in Riemannian setting.
\begin{theorem}\label{MTR}
Let $(M,g)$ be an $n$-dimensional smooth Riemannian manifold, and fix $k,p\in \mathbb{N}$ two natural numbers such that $k>0$, $p>1$, $k+p < n-1$. Let $h_1, \dots, h_p \in \mathcal{C}^{\infty}(U,\mathbb{R})$ be a given set of non-zero smooth functions defined on an open subset $U\subseteq M$, and respectively let $\{X_1,\dots, X_k, Y_1, \dots, Y_p\}\subset \mathfrak{X}(U)$ be a set of linearly independent vector fields on $U$.

Then the solutions $X\in \mathfrak{X}(U)$ of the system
\begin{equation}\label{gsi}
\left\{\begin{array}{l}
g( X, X_1)= \dots = g(X, X_k) = 0,\\
g( X, Y_1) = h_1, \dots, g(X, Y_p) = h_p,\\
\end{array}\right.
\end{equation}
are given by $X = X_{0} + X_{\perp}$, where
\begin{equation*}
X_{0}=\left\| \bigwedge_{i=1}^{p} Y_i\wedge\bigwedge_{j=1}^{k} X_j \right\|_{k+p}^{-2}\cdot\sum_{i=1}^{p}(-1)^{n-i}h_i \Theta_i, 
\end{equation*}
with $$\Theta_i = \star\left[ \bigwedge_{j=1, j\neq i}^{p} Y_j \wedge \bigwedge_{l=1}^{k} X_l  \wedge\star\left(\bigwedge_{j=1}^{p} Y_j\wedge\bigwedge_{l=1}^{k} X_l\right)\right],$$
and $X_{\perp}(x)\in \operatorname{span}_{\mathbb{R}}^{\perp_{g(x)}}\left\{ X_1 (x),\dots, X_k (x), \dots, Y_1 (x), \dots, Y_p (x)\right\}$, for any $x\in U$.
\medskip

Moreover, for each $x\in U$, there exists an open neighborhood $U_x \subseteq U$, such that for any $x^{\prime}\in U_x$
$$X_{\perp}(x^{\prime})\in \operatorname{span}_{\mathbb{R}}\left\{\star_{x^{\prime}}\left( \bigwedge_{i=1, i\neq a}^{n-(k+p)} Z_i (x^{\prime}) \wedge \bigwedge_{j=1}^{p} Y_j (x^{\prime})\wedge\bigwedge_{l=1}^{k} X_l (x^{\prime})\right): a\in\{1,\dots, n-(k+p)\}\right\},$$ where $$\{Z_1, \dots, Z_{n-(k+p)}\} \subset \mathfrak{X}(U_x) $$ is an \textbf{arbitrary} set of linearly independent vector fields on $U_x$, such that the vector fields $$\{X_1,\dots, X_k, Y_1, \dots, Y_p, Z_1, \dots, Z_{n-(k+p)}\}$$ are linearly independent on the open subset $U_x\subseteq  U$, i.e., they form a moving frame on $U_x$.
\end{theorem}
\begin{proof}
Note that for each $x\in U$, there exists an open neighborhood $U_x \subseteq U$, and a set of linearly independent vector fields $\{Z_1, \dots, Z_{n-(k+p)}\} \subset \mathfrak{X}(U_x) $, such that for any $x^{\prime}\in U_x$ 
$$
\operatorname{span}_{\mathbb{R}}\{X_1 (x^{\prime}),\dots, X_k (x^{\prime}), Y_1 (x^{\prime}), \dots, Y_p (x^{\prime}), Z_1 (x^{\prime}), \dots, Z_{n-(k+p)} (x^{\prime})\}=T_{x^{\prime}} M.
$$

Recall that the set of linearly independent vector fields $Z_1, \dots, Z_{n-(k+p)}$ with the above property, is not unique. The rest of the proof follows mimetically the proof of the Theorem \eqref{MT}.
\end{proof}

\bigskip
An immediate consequence of Proposition \eqref{aux} is the following Remark.
\begin{remark}\label{remu}
The set of vector fields $$\mathfrak{X}[X_1,\dots, X_k, Y_1, \dots, Y_p]=\{X\in\mathfrak{X}(U)\mid g(X,X_i)=g(X,Y_j)=0; 1\leq i\leq k; 1\leq j\leq p\}$$ forms an $[n-(k+p)]$-dimensional smooth distribution, locally generated around each point $x\in U$, in some open neighborhood $U_x \subseteq U$, by the set of vector fields 
$$\left\{\star\left( \bigwedge_{i=1, i\neq a}^{n-(k+p)} Z_i \wedge \bigwedge_{j=1}^{p} Y_j\wedge\bigwedge_{l=1}^{k} X_l \right): a\in\{1,\dots, n-(k+p)\right\}\subset \mathfrak{X}(U_x).$$ 

Recall that in contrast with the vector fields $X_1,\dots, X_k, Y_1, \dots, Y_p$, which are globally defined on $U$, the vector fields $Z_1 , \dots, Z_{n-(k+p)}$ are only locally defined since their existence depend on $x$, and is guaranteed in general only in some open neighborhood $U_x$ around $x$. Moreover, they are arbitrary chosen in order to be linearly independent and to complete locally the set of vector fields $\{X_1,\dots, X_k, Y_1, \dots, Y_p\}$ up to a moving frame in $U_x$.
\end{remark}
By a similar argument as in the proof of Proposition \eqref{aux}, the above defined set of local generators does not depend on the set of locally defined linearly 
independent vector fields $Z_1, \dots, Z_{n-(k+p)},$ as long as $$\{X_1,\dots, X_k, Y_1, \dots, Y_p, Z_1, \dots, Z_{n-(k+p)}\}$$ forms a moving frame.

\bigskip

Let us fix some general notations to be used for the rest of the paper. Let $\mathcal{A}\subset\mathfrak{X}(U)$ be a smooth $r$-dimensional affine distribution on the open subset $U$ of a smooth $n-$dimensional manifold $M$. This means that for each $x\in U$, there exists an open neighborhood $U_x \subseteq U$, a smooth vector field $X_0 \in \mathfrak{X}(U_x)$, and $r$ linearly independent smooth vector fields $\{X_1 ,\dots , X_r \}\subset  \mathfrak{X}(U_x)$ such that $$\mathcal{A}_{x^{\prime}} = X_0 (x^{\prime})+ \operatorname{span}_{\mathbb{R}}\{X_1 (x^{\prime}),\dots, X_{r}(x^{\prime})\},$$ for each $x^{\prime} \in U_x$. 

A set of locally defined vector fields $$\{X_0\}\biguplus \{ X_1,\dots, X_r\},$$ fulfilling the above requirements, is called \textit{a set of local generators} of the smooth affine distribution $\mathcal{A}$.

Recall that the $r-$dimensional smooth distribution that assigns to each $x\in U$ the direction of the affine space $\mathcal{A}_x$, is denoted by $L(\mathcal{A})$, and is called \textit{the linear part} of the affine distribution $\mathcal{A}$. Consequently, $L(\mathcal{A})$ can be generated locally around $x$, as
$$L(\mathcal{A})_{x^{\prime}}=\operatorname{span}_{\mathbb{R}}\{X_1 (x^{\prime}),\dots, X_{r}(x^{\prime})\},$$ for every $x^{\prime} \in U_x$. 

Note that for any fixed vector field $X\in\mathcal{A}|_{U_x}$, $$\mathcal{A}_{x^{\prime}}=X (x^{\prime})+ L(\mathcal{A})_{x^{\prime}},$$ for every $x^{\prime}\in U_x$.
\bigskip

Using the above notation for a set of local generators of a smooth affine distribution, the conclusion of the Theorem \eqref{MTR} can be reformulated as follows. 
\begin{theorem}\label{MTRD}
In the hypothesis of Theorem \eqref{MTR}, the solutions $X\in \mathfrak{X}(U)$ of the system \eqref{gsi} form the $[n-(k+p)]$-dimensional smooth affine distribution
$$
\mathfrak{A}[X_0; X_1,\dots, X_k, Y_1, \dots, Y_p]:= X_0 + \mathfrak{X}[X_1,\dots, X_k, Y_1, \dots, Y_p],
$$
locally generated by the following set of $[n-(k+p)]+1$ vector fields
$$\left\{X_0\right\}\biguplus \left\{\star\left( \bigwedge_{i=1, i\neq a}^{n-(k+p)} Z_i \wedge \bigwedge_{j=1}^{p} Y_j\wedge\bigwedge_{l=1}^{k} X_l \right): a\in\{1,\dots, n-(k+p)\right\}.$$
\end{theorem}
\begin{proof}
The proof follows by Theorem \eqref{MTR} and the fact that, in the notations of Proposition \eqref{aux}, for any $x\in U$, the vector $X(x)$ is an element of the linear variety of $T_x M$ passing through $X_0 (x)$ and having the direction $(T_{x} M)[X_1(x),\dots, X_k(x), Y_1(x), \dots, Y_p(x)]$. 
\end{proof}

\bigskip
As in the case of Theorem \eqref{MT}, let us now discuss some special cases of Theorems \eqref{MTR}, \eqref{MTRD}, namely the Riemannian analogous of Remarks \eqref{re2}, \eqref{re3}, \eqref{re1}, \eqref{re01}.

\begin{remark}\label{Re2}
\begin{itemize}
\item For $p=0$, the conclusion of Theorem \eqref{MTR} becomes as follows:

The distribution 
\begin{align*}
\mathfrak{X}[X_1,\dots, X_k]&=\{X\in\mathfrak{X}(U)\mid g(X,X_1)=\dots=g(X,X_k)=0\},
\end{align*}
is locally generated by the set of vector fields
$$\left\{\star\left( \bigwedge_{i=1, i\neq a}^{n-k} Z_i \wedge\bigwedge_{l=1}^{k} X_l \right): a\in\{1,\dots, n-k\}\right\},$$
where the set of locally defined vector fields $$\{Z_1,\dots,Z_{n-k}, X_1,\dots,X_k\}$$ forms a moving frame.

\item For $p=1$, the conclusion of Theorem \eqref{MTR} becomes as follows:

The affine distribution 
\begin{align*}
\mathfrak{A}[X_{0};X_1,\dots, X_k, Y_1]&=\{X\in\mathfrak{X}(U)\mid g(X,X_1)=\dots=g(X,X_k)=0, g(X,Y_1)=h_1\}\\
&=X_0 + \mathfrak{X}[X_1,\dots, X_k, Y_1],
\end{align*}
is locally generated by the set of vector fields
$$\left\{X_0\right\}\biguplus \left\{\star\left( \bigwedge_{i=1, i\neq a}^{n-(k+1)} Z_i \wedge Y_1 \wedge \bigwedge_{l=1}^{k} X_l \right): a\in\{1,\dots, n-(k+1)\}\right\},$$
where
$$
X_{0}=\left\| Y_1\wedge\bigwedge_{j=1}^{k} X_j \right\|_{k+1}^{-2}\cdot(-1)^{n-1}h_1 \cdot \left(\star\left[\bigwedge_{l=1}^{k} X_l \wedge\star\left(Y_1\wedge\bigwedge_{l=1}^{k} X_l\right)\right]\right), 
$$
and respectively the set of locally defined vector fields $$\{Z_1,\dots,Z_{n-(k+1)}, X_1,\dots,X_k, Y_1\}$$ forms a moving frame.
\end{itemize}
\end{remark}

\begin{remark}\label{Re3}
For $k=0$, the conclusion of Theorem \eqref{MTR} becomes as follows:

The affine distribution 
\begin{align*}
\mathfrak{A}[X_{0};Y_1,\dots, Y_p]&=\{X\in\mathfrak{X}(U)\mid g(X,Y_j)=h_j, \quad 1\leq j\leq p\}\\
&=X_0 + \mathfrak{X}[Y_1,\dots, Y_p],
\end{align*}
is locally generated by the set of vector fields
$$\left\{X_0\right\}\biguplus \left\{\star\left( \bigwedge_{i=1, i\neq a}^{n-p} Z_i \wedge \bigwedge_{j=1}^{p} Y_j \right): a\in\{1,\dots, n-p\}\right\},$$
where
$$
X_{0}=\left\|\bigwedge_{i=1}^{p} Y_i \right\|_{p}^{-2}\cdot \sum_{i=1}^{p}(-1)^{n-i}h_i \cdot \left(\star\left[\bigwedge_{j=1, j\neq i}^{p} Y_j \wedge\star\left(\bigwedge_{j=1}^{p} Y_j\right)\right]\right), 
$$
and respectively the set of locally defined vector fields $$\{Z_1,\dots,Z_{n-p}, Y_1,\dots,Y_p\}$$ forms a moving frame.
\end{remark}

\begin{remark}\label{Re1}
For $k+p=n-1$, the conclusion of Theorem \eqref{MTR} becomes as follows:

The affine distribution $\mathfrak{A}[X_0; X_1,\dots, X_k, Y_1, \dots, Y_p]$ is locally generated by the set vector fields
$$\left\{X_0\right\}\biguplus \left\{\star\left(\bigwedge_{j=1}^{p} Y_j\wedge\bigwedge_{l=1}^{k} X_l \right)\right\}.$$
\end{remark}

\begin{remark}\label{Re01}
For $k+p=n$, the conclusion of Theorem \eqref{MTR} reduces to:
\begin{align*}
\mathfrak{A}[X_0; X_1,\dots, X_k, &Y_1, \dots, Y_p]\\
&=\{X\in\mathfrak{X}(U)\mid g(X,X_i)=0, g(X,Y_j)=h_j, 1\leq i\leq k, 1\leq j\leq p\} \\
&=\{X_0\},
\end{align*}
where $$X_0 =\left\|\bigwedge_{i=1}^{p} Y_i \wedge \bigwedge_{l=1}^{k}X_l\right\|_{n}^{-2}\cdot \sum_{i=1}^{p}(-1)^{n-i}h_i \cdot \left(\star\left[\bigwedge_{j=1, j\neq i}^{p} Y_j \wedge\bigwedge_{l=1}^{k}X_l \wedge\star\left(\bigwedge_{j=1}^{p} Y_j\wedge \bigwedge_{l=1}^{k}X_l\right)\right]\right).$$

\end{remark}

\section{Applications to dynamical systems}

The aim of this section is to apply the main results from the previous section in the case of linear/affine distributions associated to conservative/dissipative dynamical systems defined eventually on an open subset $U$ of a Riemannian manifold $(M,g)$.

Before stating the main results, let us recall that a smooth function $F\in\mathcal{C}^{\infty}(U,\mathbb{R})$ is said to be a first integral (or conservation law) of the vector field $X\in\mathfrak{X}(U)$ if $\mathcal{L}_{X}F=0$, where $\mathcal{L}_{X}$ stands for the Lie derivative along the vector field $X$, or equivalently one say that $X$ conserves $F$. Similarly, a vector field $X\in\mathfrak{X}(U)$ is said to dissipate the smooth function $H\in\mathcal{C}^{\infty}(U,\mathbb{R})$ with dissipation rate $h\in\mathcal{C}^{\infty}(U,\mathbb{R})$, if $\mathcal{L}_{X}H=h$.

In the Riemannian setting, these conditions are obviously equivalent to $g(X,\nabla_{g} F)=0$, and respectively $g(X,\nabla_{g} H)=h$, where $\nabla_{g}$ stands for the gradient operator with respect to the Riemannian metric $g$.

In what follows, a vector field $X\in\mathfrak{X}(U)$ will be called \textbf{dissipative} if there exist $k,p\in\mathbb{N}$ with $k+p>0$, and a set of smooth functions $\{I_1,\dots, I_k,D_1,\dots, D_p, h_1,\dots, h_p\}\subset \mathcal{C}^{\infty}(U,\mathbb{R})$ such that the vector field $X$ conserves $I_1,\dots, I_k$ and dissipates 
$D_1,\dots, D_p$ with (corresponding) dissipation rates $h_1,\dots, h_p$. If $p=0$, the vector field $X$ will be called \textbf{conservative}.

Hence, one can apply the Theorem \eqref{MTR} in the case of linear/affine distributions associated to conservative/dissipative vector fields defined eventually on an open subset $U$ of a Riemannian manifold $(M,g)$.

\begin{theorem}\label{MTD}
Let $(M,g)$ be an $n$-dimensional smooth Riemannian manifold, and fix $k,p\in \mathbb{N}$ two natural numbers such that $k>0$, $p>1$, $k+p < n-1$. Let $h_1, \dots, h_p \in \mathcal{C}^{\infty}(U,\mathbb{R})$ be a given set of non-zero smooth functions defined on an open subset $U\subseteq M$, and respectively let $I_1,\dots, I_k, D_1,\dots, D_p\in \mathcal{C}^{\infty}(U,\mathbb{R})$ be given, such that $$\{\nabla_{g} I_1 ,\dots,\nabla_{g} I_k , \nabla_{g} D_1 , \dots,\nabla_{g} D_p\}\subset \mathfrak{X}(U)$$ form a set of linearly independent vector fields on $U$.

Then the solutions $X\in \mathfrak{X}(U)$ of the system
\begin{equation*}
\left\{\begin{array}{l}
\mathcal{L}_{X}I_1= \dots = \mathcal{L}_{X}I_k = 0,\\
\mathcal{L}_{X}D_1 = h_1, \dots, \mathcal{L}_{X}D_p = h_p,\\
\end{array}\right.
\end{equation*}
form the affine distribution (consisting of dissipative vector fields)
$$
\mathfrak{A}[X_0;\nabla_{g} I_1 ,\dots,\nabla_{g} I_k , \nabla_{g} D_1 , \dots, \nabla_{g} D_p ]=X_0 + \mathfrak{X}[\nabla_{g} I_1,\dots, \nabla_{g} I_k, \nabla_{g} D_1, \dots, \nabla_{g} D_p],
$$
locally generated by the set of vector fields
$$
\left\{X_0\right\}\biguplus \left\{\star\left( \bigwedge_{i=1, i\neq a}^{n-(k+p)} Z_i \wedge \bigwedge_{j=1}^{p} \nabla_{g} D_j \wedge\bigwedge_{l=1}^{k} \nabla_{g} I_l \right): a\in\{1,\dots, n-(k+p)\right\}
$$
where
$$
X_{0}=\left\| \bigwedge_{i=1}^{p} \nabla_{g} D_i\wedge\bigwedge_{j=1}^{k} \nabla_{g} I_j \right\|_{k+p}^{-2}\cdot\sum_{i=1}^{p}(-1)^{n-i}h_i \Theta_i, 
$$
$$
\Theta_i = \star\left[ \bigwedge_{j=1, j\neq i}^{p} \nabla_{g} D_j \wedge \bigwedge_{l=1}^{k} \nabla_{g} I_l  \wedge\star\left(\bigwedge_{j=1}^{p} \nabla_{g} D_j\wedge\bigwedge_{l=1}^{k} \nabla_{g} I_l \right)\right],
$$
and respectively the set of locally defined vector fields $$\{\nabla_{g} I_1,\dots, \nabla_{g} I_k, \nabla_{g} D_1, \dots, \nabla_{g} D_p, Z_1, \dots, Z_{n-(k+p)}\}$$ forms a moving frame.
\end{theorem}

A dynamical version of Theorem \eqref{MTD} can be formulated as follows.
\begin{theorem}
Let $\dot x =X(x)$ be the dynamical system generated by a vector field $X\in\mathfrak{X}(U)$ which conserves the smooth (functionally independent) functions $$I_1,\dots, I_k, D_1,\dots, D_p\in \mathcal{C}^{\infty}(U,\mathbb{R}).$$

Then the perturbed dynamical system 
$$
\dot x= X(x)+ X_{0}(x),  
$$
with $X_0$ given in Theorem \eqref{MTD}, is a dissipative dynamical system, generated by the dissipative vector field $X+X_0$ which conserves $I_1,\dots, I_k$, and dissipates $D_1,\dots, D_p$ with (corresponding) dissipation rates $h_1,\dots, h_p$.
\end{theorem}
\begin{proof}
The proof is a consequence of Theorem \eqref{MTD} and of the fact that $$X\in \mathfrak{X}[\nabla_{g} I_1,\dots, \nabla_{g} I_k, \nabla_{g} D_1, \dots, \nabla_{g} D_p].$$
\end{proof}

A similar conclusion was obtained in \cite{BC} for the special case $p=1$. As in the case of Theorems \eqref{MT} and respectively \eqref{MTR}, let us now give some dynamical interpretations for the analogous of Remarks \eqref{Re2}, \eqref{Re3}, \eqref{Re1}, \eqref{Re01}.

\begin{remark}\label{MTDr}
\begin{itemize}
\item For $p=0$, the conclusion of Theorem \eqref{MTD} becomes as follows:

The distribution $$\mathfrak{X}[\nabla_{g} I_1,\dots,\nabla_{g} I_k]=\{X\in\mathfrak{X}(U)\mid \mathcal{L}_{X}I_1= \dots = \mathcal{L}_{X}I_k = 0\},$$ is locally generated by the set of vector fields
$$\left\{\star\left( \bigwedge_{i=1, i\neq a}^{n-k} Z_i \wedge\bigwedge_{l=1}^{k} \nabla_{g} I_l \right): a\in\{1,\dots, n-k\}\right\},$$
where the set of locally defined vector fields $$\{Z_1,\dots,Z_{n-k}, \nabla_{g} I_l,\dots,\nabla_{g} I_k\}$$ forms a moving frame.

\item For $p=1$, the conclusion of Theorem \eqref{MTD} becomes as follows:

The affine distribution
\begin{align*}
\mathfrak{A}[X_{0};\nabla_{g} I_1,\dots, \nabla_{g} I_k, \nabla_{g} D_1]&=\{X\in\mathfrak{X}(U)\mid  \mathcal{L}_{X}I_1=\dots= \mathcal{L}_{X}I_k=0,  \mathcal{L}_{X}D_1=h_1\}\\
&=X_0 + \mathfrak{X}[\nabla_{g} I_1,\dots, \nabla_{g} I_k, \nabla_{g} D_1],
\end{align*}
is locally generated by the set of vector fields
$$\left\{X_0\right\}\biguplus \left\{\star\left( \bigwedge_{i=1, i\neq a}^{n-(k+1)} Z_i \wedge \nabla_{g} D_1 \wedge \bigwedge_{l=1}^{k} \nabla_{g} I_l \right): a\in\{1,\dots, n-(k+1)\}\right\},$$
where
$$
X_{0}=\left\| \nabla_{g} D_1 \wedge\bigwedge_{j=1}^{k} \nabla_{g} I_j \right\|_{k+1}^{-2}\cdot(-1)^{n-1}h_1 \cdot \left(\star\left[\bigwedge_{l=1}^{k} \nabla_{g} I_l \wedge\star\left(\nabla_{g} D_1 \wedge\bigwedge_{l=1}^{k} \nabla_{g} I_l \right)\right]\right), 
$$
and respectively the set of locally defined vector fields $$\{Z_1,\dots,Z_{n-(k+1)}, \nabla_{g} I_1,\dots, \nabla_{g} I_k, \nabla_{g} D_1\}$$ forms a moving frame.
\end{itemize}
\end{remark}
The first part of Remark \eqref{MTDr} (namely for $p=0$) provides a set of local generators for the distribution given by the conservative vector fields $X\in\mathfrak{X}(U)$ admitting the set of (functionally independent) first integrals $I_1,\dots, I_k \in \mathcal{C}^{\infty}(U,\mathbb{R})$. The same expression for the vector field $X_0$ as in the second part of Remark \eqref{MTDr} (namely for $p=1$) it was also found in \cite{BC}. 

\bigskip
Moreover, if $p=0$ and $k=n-1$, then the conclusion of Remark \eqref{MTDr} becomes as follows:

The vector field $\star\left(\bigwedge_{l=1}^{k} \nabla_{g} I_l \right)$
generates locally the distribution of \textbf{completely integrable vector fields} $$\mathfrak{X}[\nabla_{g} I_1,\dots,\nabla_{g} I_{n-1}]=\{X\in\mathfrak{X}(U)\mid \mathcal{L}_{X}I_1= \dots = \mathcal{L}_{X}I_{n-1} = 0\}.$$ The same conclusion was obtained also in \cite{tudoran}. 
\bigskip

\begin{remark}
For $k=0$, the conclusion of Theorem \eqref{MTD} becomes as follows:

The affine distribution
\begin{align*}
\mathfrak{A}[X_{0};\nabla_{g} D_1,\dots, \nabla_{g} D_p]&=\{X\in\mathfrak{X}(U)\mid \mathcal{L}_{X}D_j =h_j, \quad 1\leq j\leq p\}\\
&=X_0 + \mathfrak{X}[\nabla_{g} D_1,\dots, \nabla_{g} D_p],
\end{align*}
is locally generated by the set of vector fields
$$\left\{X_0\right\}\biguplus \left\{\star\left( \bigwedge_{i=1, i\neq a}^{n-p} Z_i \wedge \bigwedge_{j=1}^{p} \nabla_{g} D_j \right): a\in\{1,\dots, n-p\}\right\},$$
where
$$
X_{0}=\left\|\bigwedge_{i=1}^{p} \nabla_{g} D_i \right\|_{p}^{-2}\cdot \sum_{i=1}^{p}(-1)^{n-i}h_i \cdot \left(\star\left[\bigwedge_{j=1, j\neq i}^{p} \nabla_{g} D_j \wedge\star\left(\bigwedge_{j=1}^{p} \nabla_{g} D_j \right)\right]\right), 
$$
and respectively the set of locally defined vector fields $$\{Z_1,\dots,Z_{n-p}, \nabla_{g} D_1,\dots,\nabla_{g} D_p\}$$ forms a moving frame.
\end{remark}

\begin{remark}
For $k+p=n-1$, the conclusion of Theorem \eqref{MTD} becomes as follows:

The affine distribution $$\mathfrak{A}[X_0; \nabla_{g} I_1,\dots, \nabla_{g} I_k, \nabla_{g} D_1, \dots, \nabla_{g} D_p]$$ is locally generated by the set of vector fields $$\left\{X_0\right\}\biguplus \left\{\star\left(\bigwedge_{j=1}^{p} \nabla_{g} D_j\wedge\bigwedge_{l=1}^{k} \nabla_{g} I_l \right)\right\}.$$
\end{remark}

\begin{remark}
For $k+p=n$, the conclusion of Theorem \eqref{MTD} reduces to:
\begin{align*}
\mathfrak{A}[X_0; \nabla_{g} I_1,\dots, \nabla_{g} I_k, &\nabla_{g} D_1, \dots, \nabla_{g} D_p]\\
&=\{X\in\mathfrak{X}(U)\mid \mathcal{L}_{X}I_i =0, \mathcal{L}_{X}D_j =h_j, 1\leq i\leq k, 1\leq j\leq p\} \\
&=\{X_0\},
\end{align*}
where 
\begin{align*}
X_0 =\left\|\bigwedge_{i=1}^{p} \nabla_{g} D_i \wedge \bigwedge_{l=1}^{k} \nabla_{g} I_l\right\|_{n}^{-2}\cdot \sum_{i=1}^{p}(-1)^{n-i}h_i \Theta_i,
\end{align*}
\begin{align*}
\Theta_i =\star\left[\bigwedge_{j=1, j\neq i}^{p} \nabla_{g} D_j \wedge\bigwedge_{l=1}^{k} \nabla_{g} I_l \wedge\star\left(\bigwedge_{j=1}^{p} \nabla_{g} D_j\wedge \bigwedge_{l=1}^{k} \nabla_{g} I_l\right)\right].
\end{align*}
\end{remark}

\subsection*{Acknowledgment}
This work was supported by a grant of the Romanian National Authority for Scientific Research, CNCS-UEFISCDI, project number PN-II-RU-TE-2011-3-0103.

\bigskip
\bigskip

\noindent {\sc R.M. Tudoran}\\
West University of Timi\c soara\\
Faculty of Mathematics and Computer Science\\
Department of Mathematics\\
Blvd. Vasile P\^arvan, No. 4\\
300223 - Timi\c soara, Rom\^ania.\\
E-mail: {\sf tudoran@math.uvt.ro}\\
\medskip


\begin{thebibliography}{99}

\bibitem{BC} {\footnotesize \textsc{P. Birtea, D. Com\u{a}nescu}, Geometrical Dissipation for Dynamical Systems, \textit{Commun. Math. Phys.}, 316(2012), 375--394. }

\bibitem{Bloch} {\footnotesize \textsc{A.M. Bloch, J. Baillieul, P. Crouch, J. Marsden}, \textit{Nonholonomic Mechanics and Control}, Interdisciplinary Applied Mathematics, Springer 2003. }

\bibitem{Boothby} {\footnotesize \textsc{W.M. Boothby}, \textit{An Introduction to Differentiable Manifolds and Riemannian Geometry}, Academic Press, San Diego 2003. }

\bibitem{Calin} {\footnotesize \textsc{O. Calin, D.-C. Chang}, \textit{Sub-Riemannian Geometry: General Theory and Examples}, Encyclopedia of Mathematics and its Applications, Cambridge University Press, 2009. }

\bibitem{darling} {\footnotesize \textsc{R.W.R. Darling}, \textit{Differential forms and connections}, Cambridge University Press, 1994. }

\bibitem{jost} {\footnotesize \textsc{J. Jost}, \textit{Riemannian Geometry and Geometric Analysis}, Fifth Edition, Universitext, Springer, Berlin 2008. }

\bibitem{stefan} {\footnotesize \textsc{P. Stefan}, Accessible sets, orbits and foliations with singularities, \textit{Proc. London Math. Soc.}, 29(1974), 699--713. }

\bibitem{susman} {\footnotesize \textsc{H.J. Sussmann}, Orbits of families of vector fields and integrability of distributions, \textit{Trans. Amer. Math. Soc.}, 180(1973), 171--188. }

\bibitem{ratiurazvan} {\footnotesize \textsc{T.S. Ratiu, R.M. Tudoran, L. Sbano, E. Sousa Dias and G. Terra}, \textit{Geometric Mechanics and Symmetry: the Peyresq Lectures; Chapter II: A Crash Course in Geometric Mechanics}, pp. 23--156, London Mathematical Society Lecture Notes Series, vol. 306, Cambridge University Press 2005. }


\bibitem{tudoran} {\footnotesize \textsc{R.M. Tudoran}, A normal form of completely integrable systems, \textit{J. Geom. Phys.}, 62(5)(2012), 1167--1174. }

\end{thebibliography}
\end{document}